\documentclass[submission,copyright,creativecommons]{eptcs}
\providecommand{\event}{AFL 2026}

\usepackage{iftex}
\usepackage{amsmath,amssymb}
\usepackage{amsthm}
\usepackage{booktabs}
\usepackage{array}
\usepackage{url}
\usepackage{algorithm}
\usepackage{algpseudocode}

\ifpdf
  \usepackage{underscore}
  \usepackage[T1]{fontenc}
\else
  \usepackage{breakurl}
\fi

\title{A Non-CDCL SAT Solver with Early Conflict Detection: The Watched-Literal-Based CSFLOC Solver}
\author{G\'abor Kusper
\institute{Eszterh\'azy K\'aroly Catholic University\\ Eger, Hungary}
\email{kusper.gabor@uni-eszterhazy.hu}
}
\def\titlerunning{The Watched-Literal-Based CSFLOC Solver}
\def\authorrunning{G. Kusper}

\newtheorem{definition}{Definition}
\newtheorem{lemma}{Lemma}
\newtheorem{theorem}{Theorem}
\newtheorem{corollary}{Corollary}

\newcommand{\FLC}{\mathsf{FLC}}
\newcommand{\last}{\mathsf{last}}
\newcommand{\lastpos}{\mathsf{lastpos}}
\newcommand{\Res}{\mathsf{Res}}
\newcommand{\TO}{\mathrm{TO}}

\begin{document}
\maketitle

\begin{abstract}
CSFLOC is a non-CDCL SAT decision procedure based on counting subsumed full-length ordered clauses.
The classical CSFLOC loop traverses the ordered space of full-length clauses by a monotone counter: if the current full-length clause is not subsumed by the input formula, its negation is a satisfying assignment; otherwise, a subsuming clause determines a counter jump.
The main bottleneck is the repeated search for such a subsuming clause.
This paper presents CSFLOC-WL, and its current implementation CSFLOC-WL3, in which this search is replaced by watched-literal prefix propagation over the negation of the current full-length clause represented by the counter.
The central mechanism is early conflict detection: if propagation under a common prefix derives opposite unit consequences for the same variable, then the two reason clauses are resolved immediately and the resolvent is used as a new counter-jump cause.
The resulting solver is not a CDCL solver: it has no CDCL decision tree, no restart policy, and no first-UIP backjumping loop.
It remains a counter-guided full-length-clause-counting solver, but it imports the watched-literal data structure and reason clauses as engineering tools for discovering jumps.
Experiments on selected UNSAT SATLIB instances compare CSFLOC-WL3 with CSFLOC21TU and CaDiCaL 3.0.0.
The results are mixed: CSFLOC-WL3 is strong on several random 3-SAT instances near the random-3-SAT satisfiability threshold, whereas CSFLOC21TU remains faster on several structured cases, apparently because it contains a more mature cache mechanism that is not yet present in CSFLOC-WL3.
\end{abstract}

\section{Introduction}

Propositional satisfiability (SAT) is a central decision problem of computer science.
Most high performance complete SAT solvers today follow the conflict-driven clause learning (CDCL) paradigm.
CDCL solvers combine Boolean constraint propagation, conflict analysis, learned clauses, restarts, variable-activity heuristics, and highly optimized data structures.
The two-watched-literal technique introduced in Chaff~\cite{Moskewicz2001Chaff} became a standard implementation tool, and compact solvers such as MiniSat~\cite{EenSorensson2004MiniSat} made CDCL solver engineering accessible to a broad research community.
CaDiCaL is a recent representative of this line: it is a CDCL solver with a clean library interface and a documented architecture~\cite{Biere2024CaDiCaL}, and release 3.0.0 provides a modern reference implementation for experimental comparison~\cite{Fleury2025CaDiCaL3}.

This paper follows a different line of SAT solving.
CSFLOC is a non-CDCL solver based on counting subsumed full-length ordered clauses~\cite{KusperBiro2015CCC,KusperBiroIszaly2018CSFLOC,Kusper2020Habil,Kusper2025Engineering}.
For a fixed order of $n$ variables, an \emph{ordered full-length clause} contains exactly one literal over every variable.
A clause $D$ \emph{subsumes} a clause $C$ when $D\subseteq C$, and a CNF formula subsumes $C$ when one of its clauses does.
If a CNF formula fails to subsume some full-length clause $C$, then the negation of $C$ is a satisfying assignment; conversely, if all $2^n$ full-length clauses are subsumed, then the formula is unsatisfiable.
CSFLOC therefore decides satisfiability by traversing the ordered space of full-length clauses in a fixed variable order and by skipping intervals that are already known to be subsumed.

The original CCC algorithm traversed this space naively.
Optimized CCC used the observation that a clause whose last literal has variable index $i$ subsumes a consecutive block of $2^{n-i}$ full-length clauses.
CSFLOC added the last-positive-bit observation and organized clauses by the index and sign of their last literal~\cite{KusperBiroIszaly2018CSFLOC}.
Recent engineering work added affected clauses and length-sensitive learned-cause caches~\cite{Kusper2025Engineering,KusperSubmittedLLMCSFLOC}.
These improvements reduce scanning, but the central operation is still a subsumption test: for the current full-length clause, the solver has to find a clause that is a subset of it.

The solver studied here changes this operation.
At each iteration, the current $n$-bit counter value $b$ encodes a full-length clause $C=\FLC(b)$, and $\alpha=\overline C$ is the assignment that makes every literal of $C$ false.
Instead of directly searching for a clause $D\subseteq C$, CSFLOC-WL assigns the variables according to $\alpha$ in the fixed order and uses watched-literal propagation.
If the largest-index variable occurring in $D$ is $x_i$, then every other literal of $D$ is false after $\alpha$ has been restricted to $x_1,\ldots,x_{i-1}$. Hence $D$ is unit when variable $x_i$ is reached.
Thus the subsumption information used by CSFLOC can be recovered as a unit reason produced by prefix propagation.

The main new mechanism is \emph{early conflict detection}.
Suppose that, under a common prefix before level $i$, propagation has produced both unit consequences $x_i$ and $\bar{x}_i$ with reason clauses $P$ and $N$.
CSFLOC-WL resolves $P$ and $N$ on $x_i$ and uses the resolvent as a new CSFLOC counter-jump cause.
If the resolvent is empty, it is represented as the final counter movement to $2^n$; the abstract algorithm still reports UNSAT only when the counter has reached the end of the full-length-clause space.
If the resolvent is non-empty and its last variable is smaller than $i$, the counter may jump earlier than it would have jumped by waiting until level $i$.

The contribution of the paper is threefold.
First, it formulates CSFLOC subsumption search as watched-literal prefix propagation over the negation of the current full-length clause.
Second, it states and proves the Early Conflict Detection theorem, which turns opposite unit consequences under a common prefix into a sound CSFLOC jump cause.
Third, it reports an experimental comparison between CSFLOC-WL3, the cache-oriented CSFLOC21TU implementation, and CaDiCaL 3.0.0 on selected UNSAT SATLIB benchmarks.
The experiments show that watched-literal propagation and early conflict detection can be effective, but also that CSFLOC21TU's mature cache mechanism is still important on several benchmark families.
The paper also documents two new CSFLOC21TU renaming strategies, $T$ and $U$: $T$ prepares two-SAT-like prefixes, while $U$ tries to create early unit consequences.

It is important to delimit the claim.
CSFLOC-WL3 is not presented as a replacement for CDCL solvers.
It is a step in the development of an alternative solver family whose search space is the ordered space of full-length clauses.
It is not a classical lookahead solver either: although it assigns a prefix and observes propagation consequences, it does not branch on both values of a selected decision variable.
The prefix is fixed by the current counter value.
The purpose of propagation is not to choose a DPLL branch, but to find the next sound jump in the counter traversal.

\section{Related Work}

The standard complete SAT solver architecture is rooted in the DPLL procedure~\cite{Davis1962DPLL} and its conflict-driven extensions.
GRASP introduced systematic conflict analysis and clause learning~\cite{MarquesSilva1999GRASP}.
Chaff demonstrated the practical importance of engineering choices, in particular two-watched literals and the Variable State Independent Decaying Sum (VSIDS) branching heuristic~\cite{Moskewicz2001Chaff}.
MiniSat then showed that a compact CDCL implementation can serve both as a competitive solver and as a research platform~\cite{EenSorensson2004MiniSat}.
CSFLOC-WL borrows the watched-literal data structure from this tradition, but not the CDCL search architecture: the global control object is the CSFLOC counter over full-length clauses.

The full-length-clause-counting line started with the idea of solving SAT by an iterative version of inclusion-exclusion~\cite{KusperBiro2015CCC}.
The 2018 paper introduced CSFLOC as the next generation of full-length clause counting algorithms~\cite{KusperBiroIszaly2018CSFLOC}.
The habilitation dissertation of the author summarizes the above results and gives a detailed proof of the soundness and completeness of CSFLOC~\cite{Kusper2020Habil}.
A recent engineering paper describes CSFLOC as a subsumption-driven clause-counting SAT solver~\cite{Kusper2025Engineering}.
The present paper continues this line, but replaces the main subsumption search by watched-literal prefix propagation.

Earlier CSFLOC implementations used last-variable-indexed clause lists, affected clauses, and later bounded learned-cause caches.
In CSFLOC19 the learned-cause cache kept the most recent three clauses for each last-variable/sign pair.
CSFLOC20 admitted only short learned causes, and CSFLOC21 used a larger length-sensitive cache with short, medium, and long regions~\cite{KusperSubmittedLLMCSFLOC}.
The CSFLOC21TU implementation used in this paper belongs to this cache-oriented line.
The watched-literal implementation is intentionally less mature with
respect to caching because we encountered technical difficulties when
migrating the cache mechanism to the new implementation.
This makes the comparison informative because it separates the effect of prefix propagation and early conflict detection from the effect of previous cache engineering.

The experiments use instances from SATLIB, a benchmark collection for SAT research~\cite{HoosStutzle2000SATLIB,SATLIBBenchmarks}.
SAT\-LIB includes uniform random 3-SAT families such as the uf/uuf instances as well as classical benchmark families such as pigeonhole, Dubois, AIM, pret, bf, and ssa instances.
All instances reported in this paper are UNSAT.
The external reference solver is CaDiCaL 3.0.0.

\section{CSFLOC Preliminaries}
\label{sec:prelim}

This section recalls preliminary results on the CSFLOC SAT solver.
The original CSFLOC algorithm and its soundness and completeness
proof were introduced in earlier
work~\cite{KusperBiroIszaly2018CSFLOC,Kusper2020Habil}.

Let $V=\{x_1,\ldots,x_n\}$ be the variables of the input formula,
ordered by index.
A literal is either a variable $x_i$ or its negation $\bar{x}_i$.
A clause is a disjunction of literals, represented as a set, and a
CNF formula is a conjunction of clauses, represented as a set of
clauses.
A clause is called \emph{ordered} when its literals are written
according to the fixed variable order; this ordering does not change
its logical meaning.
A clause $C$ over $V$ is a \emph{full-length clause} if it contains
exactly one literal over every variable $x_i$.

For a non-empty clause $D$, define
\[
  \last(D)=\max\{i\mid D \text{ contains a literal over } x_i\}.
\]
For a full-length clause $C$, define
\[
  \lastpos(C)=\max(\{i\mid x_i\in C\}\cup\{1\}).
\]

During an iteration, the integer
$\mathit{count}\in\{0,\ldots,2^n-1\}$ is identified with its
$n$-bit binary representation
$b=(b_1,\ldots,b_n)\in\{0,1\}^n$, including leading zeros and with
$b_1$ as the most significant bit.
We call $\mathit{count}$, or equivalently $b$, the current
\emph{counter value}.
The value $2^n$ is the terminal counter value marking the end of the
full-length-clause space.

The bit vector $b$ represents the full-length clause
\[
  \FLC(b)=\{x_i\mid b_i=1\}\cup\{\bar{x}_i\mid b_i=0\}.
\]
The candidate assignment associated with $b$ is the complement
$\alpha(b)=\overline{\FLC(b)}$.
Thus $\ell\in\FLC(b)$ if and only if (iff) $\ell$ is false under
$\alpha(b)$.

For a clause set $S$ and a clause $D$, the notation $S\models D$
means that every truth assignment satisfying all clauses of $S$
also satisfies $D$.

\begin{definition}[Subsumption and jump cause]
A clause $D$ subsumes a clause $C$ iff $D\subseteq C$.
A clause set $S$ subsumes $C$ iff some clause
$E\in S$ subsumes $C$.
Let $S$ be the input clause set and let $b$ be the current counter
value.
A clause $D$ is a \emph{counter-jump cause} for $b$ if
$S\models D$ and $D\subseteq\FLC(b)$.
If $i=\last(D)$, then $D$ justifies the CSFLOC counter jump
determined by level $i$.
\end{definition}

\begin{lemma}[Clear Clause Rule]
\label{lem:clear}
Let $S$ be a clause set and let $C$ be a full-length clause.
Then $S$ subsumes $C$ iff $\overline{C}$ is not a model of $S$.
Equivalently, $S$ does not subsume $C$ iff $\overline{C}$ is a model of $S$.
\end{lemma}

\begin{proof}
If $S$ subsumes $C$, then some $D\in S$ satisfies $D\subseteq C$.
Every literal of $D$ is false under $\overline C$, so $\overline C$ is not a model of $S$.
Conversely, if $\overline C$ is not a model of $S$, then some $D\in S$ is false under $\overline C$.
Hence every literal of $D$ belongs to $C$, so $D\subseteq C$.
The second equivalence is the negation of the first one.
\end{proof}

\begin{corollary}[Full-length coverage]
\label{cor:coverage}
A clause set $S$ over $n$ variables is satisfiable iff it subsumes fewer than $2^n$ full-length clauses.
It is unsatisfiable iff it subsumes all $2^n$ full-length clauses.
\end{corollary}

\begin{lemma}[Consecutive block property]
\label{lem:block}
Let $D$ be a non-empty clause and let $i=\last(D)$.
If $D\subseteq\FLC(b)$, then every full-length clause with the same first $i$ literals as $\FLC(b)$ is also subsumed by $D$.
Consequently, if the suffix bits $b_{i+1},\ldots,b_n$ are all zero, then $D$ subsumes the current full-length clause and the next $2^{n-i}-1$ full-length clauses in the counter order.
\end{lemma}

\begin{proof}
All literals of $D$ have index at most $i$.
Any full-length clause with the same first $i$ literals as $\FLC(b)$ contains all literals of $D$.
There are $2^{n-i}$ possible suffixes after the first $i$ variables.
\end{proof}

\begin{algorithm}[H]
\caption{\textsc{CSFLOC}($S$)}
\label{alg:csfloc}
\begin{algorithmic}[1]
\Require $S$ is a non-empty set of ordered clauses over variables $x_1,\ldots,x_n$.
\Ensure If $S$ is satisfiable, return a model of $S$; otherwise return the empty set.
\State $S_i \gets \{D\in S\mid \last(D)=i\}$ for $i=1,\ldots,n$
\State $\mathit{count}\gets 0$
\While{$\mathit{count}<2^n$}
    \State $\mathit{increment}\gets 0$
    \State $C\gets \FLC(\mathit{count})$
    \For{$j\gets \lastpos(C)$ \textbf{to} $n$}
        \If{there exists $D\in S_j$ such that $D\subseteq C$}
            \State $\mathit{increment}\gets 2^{n-j}$
            \State \textbf{break}
        \EndIf
    \EndFor
    \If{$\mathit{increment}=0$}
        \State \Return $\overline{C}$
    \Else
        \State $\mathit{count}\gets \mathit{count}+\mathit{increment}$
    \EndIf
\EndWhile
\State \Return $\{\}$
\end{algorithmic}
\end{algorithm}

The only difference between Optimized CCC and CSFLOC is the starting point of the inner search.
Optimized CCC starts at $j=1$, while CSFLOC starts at $\lastpos(C)$.
For a fixed current clause $C$, a smaller value of $j$ yields a larger block of size $2^{n-j}$.
By the \emph{best available jump} we mean the jump supplied by an applicable subsuming clause with the smallest possible last-variable index.
Because Algorithm~\ref{alg:csfloc} scans $j$ in increasing order, it selects this largest available block.
The last-positive-bit observation states that, provided earlier iterations also selected their best available jumps, no clause with last variable below $\lastpos(C)$ can subsume the current full-length clause~\cite{KusperBiroIszaly2018CSFLOC,Kusper2020Habil}.
A subsuming clause \emph{justifies} its jump because Lemma~\ref{lem:block} guarantees that every full-length clause skipped before the next block boundary is also subsumed; therefore the jump cannot skip a satisfying assignment.
Algorithm~\ref{alg:csfloc} is consequently sound and complete: if no subsuming clause is found, Lemma~\ref{lem:clear} gives a model; if a subsuming clause is found, the block lemma validates the movement; and if the counter reaches $2^n$, Corollary~\ref{cor:coverage} gives unsatisfiability.

\section{Watched-Literal Prefix Propagation}
\label{sec:theory}

The purpose of CSFLOC-WL is to replace the explicit subsumption search in Algorithm~\ref{alg:csfloc} by watched-literal prefix propagation under $\alpha(b)=\overline{\FLC(b)}$.
The global structure remains the CSFLOC counter loop.
The solver still traverses the ordered space of full-length clauses, and the only UNSAT return is the classical CSFLOC one: the counter reaches the end of the full-length-clause space.
The difference is how the next counter increment is found.

\subsection{Terminology of Watched-Literal Prefix Propagation}
For $i\in\{1,\ldots,n\}$, let $\alpha_{<i}(b)$ denote the
restriction of $\alpha(b)$ to the variables
$x_1,\ldots,x_{i-1}$.
Assignments taken directly from this restriction are called
\emph{prefix decisions}.
In an explanation clause, see the next paragraph, a \emph{prefix-decision literal} is a
literal falsified by one of these assignments rather than by a
propagated assignment.

A partial assignment assigns truth values to some of the variables.
A clause $D$ is \emph{unit under a partial assignment $A$ with unit
literal $u$} if $u\in D$ is unassigned and every literal in
$D\setminus\{u\}$ is false under $A$.
Unit propagation then assigns $u$ to true; $u$ is called a
\emph{propagated literal}, and $D$ is its \emph{reason clause}.
A \emph{propagation conflict} occurs when every literal of some
clause is false.
For the current full-length clause $C$, a propagated literal $u$ is
\emph{relevant} if $u\in C$, that is, if the counter-complement
assignment $\overline C$ would make $u$ false.
An \emph{explanation clause} for a propagated literal or a conflict
is obtained by resolving reason clauses until only 
the relevant propagated literal, if any, and prefix-decision literals remain.

The \emph{watched-literal data structure for $S$} associates each
clause with one or two watched literals and maintains, for each
literal, a list of clauses that currently watch it.
When a watched literal becomes false, only the clauses on its watch
list need to be inspected.
A clause then moves the watch to another non-false literal if
possible; otherwise it is already satisfied, becomes unit and
propagates its remaining literal, or becomes conflicting.

\begin{definition}[Prefix-valid cause]
\label{def:prefix-valid}
Let $b$ be the current counter value and let $D$ be a non-empty
clause with $S\models D$.
Let $i=\last(D)$ and let $p_i(D)$ be the literal of $D$ over $x_i$.
We say that $D$ is \emph{prefix-valid at level $i$ for $b$} if
\[
  D\setminus\{p_i(D)\}\subseteq\FLC(b).
\]
The literal $p_i(D)$ is called the pivot literal of $D$ at level $i$.
\end{definition}

A prefix-valid cause becomes an ordinary counter-jump cause exactly
when its pivot literal also belongs to $\FLC(b)$.
If the pivot literal does not belong to $\FLC(b)$, the clause is still
useful as a reason for a forced value opposite to the current counter
branch.
Two prefix-valid causes at level $i$ are called \emph{opposite} if
their pivot literals are $x_i$ and $\bar{x}_i$, respectively.

Here level $j$ refers to the position of variable $x_j$ in the fixed
variable order.
For $B=2^{n-j}$ and
$q=\lfloor\mathit{count}/B\rfloor$, the \emph{level-$j$ block}
containing $\mathit{count}$ is the interval
\[
  \{qB,\ldots,(q+1)B-1\}.
\]
Its counter values have the same first $j$ bits.
The \emph{next boundary} of this block is $(q+1)B$, so the distance
from the current counter value to that boundary is
$B-(\mathit{count}\bmod B)$.


\subsection{The CSFLOC SAT Solver with Watched-Literals Data Structure}

\begin{algorithm}[H]
\caption{\textsc{CSFLOC-WL}($S$)}
\label{alg:csfloc-wl}
\begin{algorithmic}[1]
\Require $S$ is a non-empty set of ordered clauses over variables $x_1,\ldots,x_n$.
\Ensure If $S$ is satisfiable, return a model of $S$; otherwise return the empty set.
\State $\mathit{count}\gets 0$
\State initialize the watched-literal data structure for $S$
\While{$\mathit{count}<2^n$}
    \State $\mathit{increment} \gets 0$
    \State $C \gets \FLC(\mathit{count})$
    \State $\mathit{start}\gets \lastpos(C)$
    \State propagate $\overline{C}$ by watched literals, using $\mathit{start}$ as the first ordinary CSFLOC search level, until one of the following events occurs:
    \Statex \hspace{\algorithmicindent}(a) a relevant unit $u\in C$ is derived, with explanation clause $D$;
    \Statex \hspace{\algorithmicindent}(b) a conflict is derived, with explanation clause $D$;
    \Statex \hspace{\algorithmicindent}(c) opposite prefix-valid causes $P$ and $N$ yield a non-tautological resolvent $D$;
    \Statex \hspace{\algorithmicindent}(d) propagation finishes without a relevant unit, conflict, or non-tautological complementary resolvent.
    \If{event (d) occurs}
        \State \Return $\overline{C}$
    \EndIf
    \If{$D=\emptyset$}
        \State $\mathit{increment}\gets 2^n-\mathit{count}$
    \Else
        \State $j\gets \last(D)$
        \State $B\gets 2^{n-j}$
        \State $\mathit{increment}\gets B-(\mathit{count}\bmod B)$
    \EndIf
    \State $\mathit{count}\gets \mathit{count}+\mathit{increment}$
\EndWhile
\State \Return $\{\}$
\end{algorithmic}
\end{algorithm}

The operation ``propagate $\overline{C}$ by watched literals'' is intentionally abstract in Algorithm~\ref{alg:csfloc-wl}.
It denotes standard unit propagation under the counter-complement assignment.
The implementation details of watched lists, trails, queues, and reason arrays are not part of the abstract algorithm.
What matters for CSFLOC-WL is only the event returned by propagation.
A relevant unit is a propagated literal $u$ such that $u\in C$.
Equivalently, $u$ is a literal that the candidate assignment $\overline C$ would make false.
A conflict explanation is obtained from a propagation conflict by resolving reason clauses until only prefix-decision literals remain.
A complementary resolvent is obtained when two opposite prefix-valid causes meet at the same variable level.
Tautological complementary resolvents are ignored, because they exclude no full-length clause.

The formula $\mathit{increment}=B-(\mathit{count}\bmod B)$ moves the counter to the next boundary of the level-$j$ block.
In the classical CSFLOC case, where $j$ is not below the last positive bit, this is exactly the usual increment $2^{n-j}$.
The more general form is needed because early conflict detection may produce a cause whose last variable is below the current CSFLOC start index.
If $D=\emptyset$, the algorithm uses the final increment $2^n-\mathit{count}$; thus an empty derived clause is represented as a counter movement to the end of the full-length-clause space.


\begin{lemma}[Subsuming clause gives a prefix unit]
\label{lem:subsumption-unit}
Let $C=\FLC(b)$ be the current full-length clause, and let $D\subseteq C$ be a non-empty clause.
Let $i=\last(D)$, and let $u$ be the unique literal of $D$ over $x_i$.
Then $D$ is unit under $\alpha_{<i}(b)$ with unit literal $u$.
\end{lemma}

\begin{proof}
Every literal of $D\setminus\{u\}$ has index smaller than $i$ and belongs to $\FLC(b)$.
Hence it is false under $\alpha_{<i}(b)$.
The literal $u$ is over $x_i$, so it is not assigned by $\alpha_{<i}(b)$.
\end{proof}

\begin{lemma}[Relevant prefix unit gives a jump cause]
\label{lem:unit-jump}
Let $D$ be a clause entailed by $S$.
Assume that $D$ is unit under $\alpha_{<i}(b)$ with unit literal $u$, where $u$ is over $x_i$ and $i=\last(D)$.
If $u\in\FLC(b)$, then $D\subseteq\FLC(b)$, and therefore $D$ is a valid counter-jump cause for $b$.
\end{lemma}

\begin{proof}
All non-unit literals of $D$ are false under $\alpha_{<i}(b)$, hence belong to $\FLC(b)$.
The unit literal $u$ belongs to $\FLC(b)$ by assumption.
Thus $D\subseteq\FLC(b)$, and the entailment $S\models D$ makes $D$ a counter-jump cause.
\end{proof}

Explanation clauses used in events (a) and (b) are obtained from reason clauses by repeated resolution.
Reason clauses are original input clauses or clauses previously derived by sound resolution steps.
Consequently, every non-empty explanation clause returned by event (a) or (b) is entailed by $S$ and is a subset of $\FLC(b)$: in the relevant-unit case it contains the relevant unit and prefix-decision literals, and in the conflict case it contains only prefix-decision literals.
Hence every non-empty explanation is a valid counter-jump cause.
If the explanation is empty, Algorithm~\ref{alg:csfloc-wl} represents it by the final increment $2^n-\mathit{count}$.



\begin{theorem}[Early Conflict Detection]
\label{thm:early-conflict-detection}
Let $S$ be a CNF formula and let $b$ be the current CSFLOC counter value.
Let $i\in\{1,\ldots,n\}$.
Assume that $P$ and $N$ are two prefix-valid causes at level $i$ for $b$ such that $x_i\in P$ and $\bar{x}_i\in N$.
Define
\[
  D=\Res_i(N,P)=(N\cup P)\setminus\{\bar{x}_i,x_i\},
\]
with duplicate literals removed.
If $D$ is tautological, it is ignored.
If $D$ is empty, CSFLOC-WL represents it as the final increment to $2^n$.
If $D$ is non-empty and non-tautological, then $S\models D$, $D\subseteq\FLC(b)$, and $\last(D)<i$; hence $D$ is a valid counter-jump cause.
\end{theorem}

\begin{proof}
Since $P$ and $N$ are prefix-valid causes, $S\models P$ and $S\models N$.
By soundness of resolution, $S\models D$ whenever the resolvent is considered as a non-tautological clause.
If $D$ is tautological, it excludes no full-length clause and is ignored.
If $D$ is empty, it is contained in every full-length clause and is represented by the final increment.
Now assume that $D$ is non-empty and non-tautological.
Prefix validity gives $P\setminus\{x_i\}\subseteq\FLC(b)$ and $N\setminus\{\bar{x}_i\}\subseteq\FLC(b)$.
The resolvent contains only literals from these two non-pivot parts, so $D\subseteq\FLC(b)$.
All non-pivot literals have index strictly smaller than $i$, hence $\last(D)<i$.
\end{proof}

\begin{corollary}[Early jump]
\label{cor:early-jump}
In the non-empty non-tautological case of Theorem~\ref{thm:early-conflict-detection}, CSFLOC-WL may jump according to $\last(D)$.
Since $\last(D)<i$, this jump is earlier than the level at which the opposite unit consequences meet.
If $\last(D)$ is also below the current CSFLOC start index, then the jump skips a larger block than the ordinary last-positive-bit search would consider in that iteration.
\end{corollary}

\begin{theorem}[Soundness and completeness of CSFLOC-WL]
\label{thm:wl-sound-complete}
Assume that watched-literal propagation in Algorithm~\ref{alg:csfloc-wl} implements standard unit propagation under the current prefix assignment, and that explanation clauses are built by sound resolution steps.
Then Algorithm~\ref{alg:csfloc-wl} is sound and complete for propositional satisfiability.
\end{theorem}

\begin{proof}
Every non-zero increment is justified by a clause entailed by $S$ and contained in the current full-length clause $C$, or by the empty clause represented as the final increment.
For event (a), this follows from Lemma~\ref{lem:unit-jump}; for event (b), from the explanation invariant; for event (c), from Theorem~\ref{thm:early-conflict-detection}.
Thus no counter increment skips a satisfying assignment.
If propagation finishes with event (d), then no input clause subsumes $C$: otherwise Lemma~\ref{lem:subsumption-unit} would make such a clause appear as a relevant unit, or propagation would derive a conflict.
By Lemma~\ref{lem:clear}, $\overline C$ is a model of $S$.
The only UNSAT return occurs after the loop, when the counter has covered the full-length-clause space by sound increments; Corollary~\ref{cor:coverage} then gives unsatisfiability.
Termination follows because every non-returning iteration increases the monotone counter by a positive amount.
\end{proof}

CSFLOC-WL uses watched literals, unit reasons, and resolution, but it does not use the CDCL search architecture.
A CDCL solver would interpret an implication conflict as a reason to learn a clause and backjump in a decision tree.
CSFLOC-WL uses the same local propagation information to move a monotone full-length-clause counter.
The learned or derived object is therefore a CSFLOC jump cause, not primarily a first-UIP backjump clause.

\section{Implementation Notes}
\label{sec:impl}

The implementation evaluated in this paper is the Java prototype CSFLOC-WL3 available at \url{https://uni-eszterhazy.hu/fmv/tools}.
It has four layers: input normalization and variable renaming; ordered clause and counter representation; watched-literal prefix propagation; and the CSFLOC control loop.
The input is read in DIMACS CNF format.
After parsing, duplicate literals are removed, tautological clauses are ignored, and clauses are sorted according to the selected variable order.
The full-length-clause counter is represented as a bit vector with the same semantics as in Section~\ref{sec:prelim}.

The propagation layer uses the standard two-watched-literal invariant.
A clause is inspected only when one of its watched literals becomes false.
If a replacement watched literal can be found, the clause is not processed further.
If not, the other watched literal either becomes a propagated unit or yields a conflict.
The solver stores a propagation trail and reason clauses for propagated variables.
When the prefix state is reset for the next counter value, assignments and reasons are cleared, but the watched positions are not restored; unassigning variables cannot invalidate the watch invariant.

The main loop is a hybrid.
It first tries known active causes; on a miss, it runs watched-literal prefix propagation and converts relevant units, conflicts, and complementary causes into counter-jump causes.
CSFLOC-WL3 stores the active positive and negative causes needed by the traversal and by early conflict detection, but it does not yet have the mature CSFLOC21TU learned-cause cache.
This design decision is important for the experiments: some regressions are expected when a solver that discovers good causes by propagation is compared with a solver that already has an efficient mechanism for remembering and reusing them.

\subsection{Variable-renaming strategies T and U}

CSFLOC21TU, available at \url{https://uni-eszterhazy.hu/fmv/tools}, and also CSFLOC-WL3 contain two additional variable-renaming strategies, denoted by $T$ and $U$.
They belong to the same High\-Level\-Reader-style preprocessing mechanism as the earlier $B,H,C,W,R,I,$ and $S$ strategies known from CSFLOC19.
A strategy string is applied from left to right: after each strategy a translation table is computed, the clauses are renamed, and the new translation is composed with the previous translations.
Thus a label such as CUB denotes a sequential composition of clustering, island-oriented renaming, and unit-oriented renaming.
The strategies change only the variable order; they do not change the formula.

The $T$ strategy is a two-SAT-prefix-oriented renaming strategy.
Its goal is to give small indices to a greedy prefix of variables such that, after this prefix has been fixed, long clauses are reduced to clauses of length at most two on as many counter branches as possible.
For 3-SAT this becomes a hitting-set-like task: the selected prefix should intersect every 3-clause.
The strategy uses a waiting-list discipline: after a variable has been selected, related variables from still-uncovered long clauses are temporarily delayed, so the prefix spreads over many clauses instead of concentrating in one local region.
At the moment CSFLOC has no dedicated 2-SAT watcher or integrated 2-SAT solver, so $T$ is mainly preparation for a future extension.

The $U$ strategy uses the same waiting-list philosophy, but its target is directly connected to propagation.
It tries to arrange the variable order so that unit clauses appear early and frequently under counter-induced prefixes.
For a clause of length $k$, it tries to place $k-1$ useful variables early, so that the remaining literal can become unit.
This is especially natural for CSFLOC-WL, where jump causes are discovered through prefix units and early conflicts.
In the reported experiments, the CUB strategy is used on the two ssa instances; these rows are not an ablation study of $U$, but they show that a combination containing $U$ was among the selected strong CSFLOC configurations.

\section{Experimental Evaluation}
\label{sec:experiments}

The experiments compare three solvers.
CSFLOC21TU is the previous CSFLOC implementation extended with the new $T$ and $U$ variable-renaming strategies and with the mature cache mechanism inherited from CSFLOC21.
CSFLOC-WL3 is the new watched-literal-based implementation with early conflict detection.
All CSFLOC variants are available at \url{https://uni-eszterhazy.hu/fmv/tools}.
CaDiCaL 3.0.0 is used as a modern CDCL reference solver.
All test instances are UNSAT and were downloaded from SATLIB in DIMACS CNF format.
The selected instances cover several SATLIB families: uuf random unsatisfiable 3-SAT, pigeonhole, Dubois, pret, ssa, bf, and AIM instances, see: \url{https://www.cs.ubc.ca/~hoos/SATLIB/benchm.html}. 
The goal is not a complete benchmark study over all instances of each family, but an initial engineering comparison.

The experiments were executed on a laptop with a 12th Gen Intel(R) Core(TM) i7-1255U CPU at 1.70 GHz and 32 GB of memory.
Each selected instance was executed three times for each solver, and the table reports the arithmetic mean of the three wall-clock running times in seconds.
The notation $\TO$ means that the run did not finish within the timeout limit of 120 seconds.
Because several measurements are below one tenth of a second, very small differences should not be over-interpreted.

\begin{table}[t]
\centering
\scriptsize
\setlength{\tabcolsep}{3pt}
\begin{tabular}{lrrlrrrrr}
\toprule
Instance & $n$  & $m$ & Strategy & C19 & C21TU & C21NoCache & CWL3 & CaDi \\
\midrule
uuf50-01        & 50   & 218  & IWCR &  0.017 &  0.031 & 0.018 & 0.004 & 0.03 \\
uuf75-01        & 75   & 325  & IWCR &  0.053 &  0.054 & 0.061 & 0.012 & 0.05 \\
uuf100-01       & 100  & 430  & IWCR &  0.135 &  0.146 & 0.191 & 0.025 & 0.05 \\
uuf125-01       & 125  & 538  & IWCR &  1.099 &  0.436 & 1.746 & 0.071 & 0.09 \\
uuf150-01       & 150  & 645  & IWCR & 52.156 & 23.929 & $\TO$ & 1.179 & 0.14 \\
uuf175-01       & 175  & 753  & CIBR &  $\TO$ &  $\TO$ & $\TO$ & 6.856 & 0.31 \\
uuf200-01       & 200  & 860  & CIBR &  $\TO$ &  $\TO$ & $\TO$ & 24.825 & 0.79 \\
hole7           & 56   & 204  & B    &  0.053 &  0.063 & 0.087 & 0.289 & 0.03 \\
hole8           & 72   & 297  & B    &  0.296 &  0.208 & 0.647 & 3.812 & 0.07 \\
hole9           & 90   & 415  & B    &  3.333 &  2.236 & 6.885 & 65.269 & 0.05 \\
dubois30        & 90   & 240  & CIBR &  0.475 &  1.187 & 11.346 & 7.009 & 0.02 \\
pret60\_75      & 60   & 160  & NON  &  0.052 &  0.071 & 0.058 & 0.105 & 0.05 \\
ssa0432-003     & 435  & 1027 & CUB  &  0.227 &  0.242 & 3.483 & 0.372 & 0.04 \\
ssa2670-141     & 986  & 2315 & CUB  &  0.791 &  0.575 & 31.054 & 4.379 & 0.06 \\
bf2670-001      & 1393 & 3434 & IWCR &  0.082 &  0.093 & 0.083 & 0.043 & 0.07 \\
bf0432-007      & 1040 & 3668 & IWCR & 79.435 & 87.233 & $\TO$ & $\TO$ & 0.06 \\
aim-200-2\_0-no-1 & 200 & 400 & IWCR &  0.015 &  0.043 & 0.015 & 0.042 & 0.05 \\
aim-200-2\_0-no-2 & 200 & 400 & IWCR &  0.018 &  0.038 & 0.021 & 0.034 & 0.05 \\
\bottomrule
\end{tabular}
\caption{Runtime comparison on selected UNSAT SATLIB instances. Times are arithmetic means of three runs, in seconds. C19, C21TU, CWL3, C21NoCache, and CaDi abbreviate CSFLOC19, CSFLOC21TU, CSFLOC21TU with no cache, CSFLOC-WL3, and CaDiCaL 3.0.0, respectively.}
\label{tab:runtime}
\end{table}

To interpret Table~\ref{tab:runtime}, it is useful first to clarify the strategy labels and the main differences among the CSFLOC variants.
The strategy column records the variable-renaming strategy sequence used by the CSFLOC variants.
The same strategy is used for CSFLOC19, CSFLOC21TU, C21NoCache, and CSFLOC-WL3 in each row; CaDiCaL does not use this CSFLOC-specific parameter.
CSFLOC19 is regarded as the most stable CSFLOC version, so it serves as the baseline within the CSFLOC solver family.
Compared with CSFLOC19, CSFLOC21TU includes a new learned-cause mechanism~\cite{KusperSubmittedLLMCSFLOC} and the two new variable-renaming strategies, \(T\) and \(U\).
For each last-variable/sign pair, CSFLOC19 maintains a three-entry cache containing the three most recently learned clauses with that last literal.
CSFLOC21TU inherits the length-sensitive cache of CSFLOC21, which stores up to ten learned clauses for each last literal: three short clauses of length at most 5, three medium clauses of length 6--9, and four long clauses of length at least 10.
Apart from these extensions, CSFLOC21TU follows the same underlying CSFLOC architecture as CSFLOC19.

The rows using the CUB strategy are relevant to the new U strategy. Since CSFLOC19 does not implement U, it applies the corresponding strategy sequence without that component, namely CB.

The first observation is that CSFLOC-WL3 outperforms CSFLOC21TU without caching on most instances, with the main exceptions being the pigeonhole instances and a few instances with very short running times. This suggests that the watched-literal data structure is useful not only for CDCL-based SAT solvers, but also for CSFLOC-based solvers.

The first group of results is the random 3-SAT uuf family.
Here CSFLOC-WL3 is consistently faster than CSFLOC21TU on the reported instances.
On uuf150-01 the runtime decreases from 23.929 seconds to 1.179 seconds, and CSFLOC-WL3 solves uuf175-01 and uuf200-01 while CSFLOC21TU reaches the timeout limit.
This is the clearest evidence that watched-literal prefix propagation and early conflict detection can be useful for the CSFLOC family.

The second group shows the limitation of the current WL3 prototype.
On hole7, hole8, hole9, dubois30, pret60\_75, ssa2670-141, and bf0432-007, CSFLOC21TU is faster.
These results support the engineering hypothesis that the mature cache mechanism of CSFLOC21TU is still important, or even the least mature cache mechanism of CSFLOC19 is enough in most of the cases.
CSFLOC-WL3 has early conflict detection, but it does not yet have an equally developed cache for reusing derived causes across counter states.

The comparison with CaDiCaL should be interpreted as a reference point rather than as a claim of general competitiveness.
CaDiCaL is faster on most of the reported instances, particularly on the structured pigeonhole, Dubois, ssa, and bf0432 cases.
Nevertheless, CSFLOC-WL3 is faster than CaDiCaL on several selected cases, including uuf50-01, uuf75-01, uuf100-01, uuf125-01, and bf2670-001.
The two AIM rows also slightly favor CSFLOC-WL3, but the differences are too small to support a separate conclusion.
Thus, the data do not show that CSFLOC-WL3 competes with state-of-the-art CDCL solvers as a general-purpose SAT solver.
Rather, they show that the watched-literal CSFLOC direction is viable enough to justify further engineering.

\subsection{Internal statistics explaining the runtime differences}

Table~\ref{tab:internal-diagnostics} gives internal statistics for selected representative rows.
The most important quantity is the number of main-loop iterations.
Since CSFLOC is counter-guided, a large reduction in the number of main-loop iterations usually means that the solver traverses the full-length-clause space by larger jumps.
Conversely, if CSFLOC-WL3 performs more main-loop iterations than CSFLOC21TU, then the cost of watched-literal propagation and explanation construction is unlikely to be compensated.

\begin{table}[t]
\centering
\scriptsize
\setlength{\tabcolsep}{2.5pt}
\begin{tabular}{lrrrrrr}
\toprule
Instance & C21TU loops & CWL3 loops & Loop fac. & C21TU LH & CWL3 CH & CWL3 below \\
\midrule
uuf100-01    & 164\,004      & 1\,223       & 134.1 & 47\,549      & 852         & 381 \\
uuf150-01    & 46\,244\,244  & 322\,174     & 143.5 & 14\,838\,523 & 125\,857    & 70\,217 \\
uuf200-01    & $\TO$         & 4\,947\,939  & --    & $\TO$        & 2\,222\,147 & 1\,306\,067 \\
hole9        & 6\,733\,207   & 49\,816\,953 & 0.14  & 1\,466\,401  & 5\,214\,584 & 3\,439\,567 \\
dubois30     & 1\,218\,632   & 14\,105\,280 & 0.09  & 217\,918     & 1\,124\,287 & 562\,175 \\
ssa0432-003  & 689\,171      & 5 156\,235   & 4.41  & 81\,286      & 10\,184     & 4\,422 \\
ssa2670-141  & 2\,646\,742   &  1\,246\,814 & 2.12  & 205\,027     & 80\,405     &  37\,572  \\
bf2670-001   & 65\,990       & 1\,457       & 45.3  & 352          & 26          & 13 \\
bf0432-007   & 308\,005\,246 & $\TO$        & --    & 40\,375\,909 & $\TO$       & $\TO$ \\
\bottomrule
\end{tabular}
\caption{Internal diagnostic statistics. Loop fac. is C21TU loops divided by CWL3 loops. C21TU LH denotes CSFLOC21TU learned-clause hits. CWL3 CH denotes CSFLOC-WL3 Early-Conflict-Detection-theorem hits. CWL3 below denotes Early jumps below the current start index.}
\label{tab:internal-diagnostics}
\end{table}

The first pattern is visible on the random uuf instances.
On uuf100-01 and uuf150-01, CSFLOC-WL3 reduces the number of main-loop iterations by more than two orders of magnitude.
For example, uuf100-01 decreases from 164\,004 iterations in CSFLOC21TU to 1\,223 iterations in CSFLOC-WL3.
This reduction is accompanied by 852
Early-Conflict-Detection-theorem hits, including 381 
Early jumps below the current start index.
These numbers support the interpretation that CSFLOC-WL3 is fast on the uuf rows because prefix propagation and early conflict detection produce lower-level jump causes and therefore skip larger blocks of the full-length-clause space.

The second pattern explains the regressions.
On hole9, and dubois30 CSFLOC-WL3 performs substantially more main-loop iterations than CSFLOC21TU.
In these cases the watched-literal mechanism does not shorten the counter traversal enough.
Moreover, a WL3 iteration may include propagation, trail scanning, and explanation construction.
Thus, when the number of WL3 iterations is larger, the propagation overhead becomes visible.

The third pattern concerns caching.
CSFLOC21TU has a mature learned-cause cache, and the diagnostic table shows that this cache is used heavily on several structured instances.
For example, CSFLOC21TU has 1\,466\,401 learned-clause hits on hole9 and 40\,375\,909 learned-clause hits on bf0432-007.
The reported CSFLOC-WL3 runs, in contrast, use no learned-cause cache.
They rely on active causes, watched-literal prefix propagation, and complementary resolution.
This explains why WL3 can win when these mechanisms greatly shorten the traversal, but loses when CSFLOC21TU's learned-cause cache is more effective.

Overall, Tables~\ref{tab:runtime} and~\ref{tab:internal-diagnostics} suggest that CSFLOC-WL3 wins 
when watched-literal prefix propagation and complementary resolution reduce the number of counter states.
It loses when this reduction is absent, because propagation is more expensive than a successful cached-cause lookup.
The next engineering step is therefore clear: 
combine the WL3 propagation and early-conflict mechanisms with a CSFLOC21TU-level learned-cause cache.

\section{Conclusion}

This paper presented CSFLOC-WL and its current implementation CSFLOC-WL3, a non-CDCL SAT solver that replaces the classical CSFLOC subsumption search by watched-literal prefix propagation.
The theoretical bridge is simple: a clause subsuming the current full-length clause becomes a unit clause under the negated prefix assignment.
The main new rule is early conflict detection.
If opposite unit consequences are produced for the same variable under a common prefix, their reason clauses can be resolved immediately, and the resolvent is a sound CSFLOC counter-jump cause.
This uses a CDCL-style propagation event, but not the CDCL control architecture.
The global search remains a monotone traversal of the ordered full-length-clause space.

The experiments show both promise and a clear engineering gap.
CSFLOC-WL3 performs very well on several random 3-SAT instances and solves two uuf instances on which CSFLOC21TU reaches the timeout limit.
On several structured benchmarks, however, CSFLOC21TU is faster, apparently because its advanced cache mechanism reuses derived causes more effectively.
The current evidence supports a combined direction: add a CSFLOC21TU-style cache to the watched-literal implementation and then reevaluate the solver on a larger SATLIB subset.
The new $T$ and $U$ variable-renaming strategies are part of the current engineering picture.
Our broader experiments suggest that $U$ can be useful, especially in CUB-like combinations, while $T$ is mainly a preparation for future explicit 2-SAT handling.

\section*{Acknowledgment}
This project, i.e., Project no. 2024-1.2.5-TÉT-2024-00072 has been implemented with the support provided by the Ministry of Culture and Innovation of Hungary from the National Research, Development and Innovation Fund, financed under the 2024-1.2.5-TÉT funding scheme.

This work was also supported by the Slovak Research and Development Agency under the Contract no. SK-HU-24-0037.

\bibliographystyle{eptcs}
\bibliography{csfloc_wl3_final_v2.bib}
\end{document}